\documentclass[12pt]{article}

\usepackage{fullpage}
\usepackage{amsfonts,amsmath,epsf,epsfig,bbm}
\usepackage{amsthm}

\title{Minimum Tournaments with the Strong $S_k$-Property \\
       and Implications for Teaching}

\author{Hans Ulrich Simon}

\newtheorem{theorem}{Theorem}[section]
\newtheorem{definition}[theorem]{Definition}

\newtheorem{corollary}[theorem]{Corollary}
\newtheorem{remark}[theorem]{Remark}

\newcommand{\sm}{\setminus}

\newcommand{\Fp}{{\mathbbm F}}

\newcommand{\cP}{{\mathcal P}}
\newcommand{\cX}{{\mathcal X}}
\newcommand{\cC}{{\mathcal C}}

\newcommand{\TD}{\mathrm{TD}}
\newcommand{\RTD}{\mathrm{RTD}}
\newcommand{\NCTD}{\mathrm{NCTD}}

\newcommand{\seq}{\subseteq}
\newcommand{\ra}{\rightarrow}
\newcommand{\impl}{\Rightarrow}

\begin{document}

\maketitle

\begin{abstract}
A tournament is said to have the $S_k$-property if, for any set
of $k$ players, there is another player who beats them all.
Minimum tournaments having this property have been explored very 
well in the 1960's and the early 1970's. In this paper, we define
a strengthening of the $S_k$-property that we name ``strong
$S_k$-property''. We show, first, that several basic results on 
the weaker notion remain valid for the stronger notion (and the
corresponding modification of the proofs requires only little 
extra-effort). Second, it is demonstrated that the stronger 
notion has applications in the area of Teaching. Specifically, 
we present an infinite family of concept classes all of which 
can be taught with a single example in the No-Clash model of 
teaching while, in order to teach a class $\cC$ of this family 
in the recursive model of teaching, order of $\log|\cC|$ many 
examples are required. This is the first paper that presents 
a concrete and easily constructible family of concept classes 
which separates the No-Clash from the recursive model of teaching 
by more than a constant factor. The separation by a logarithmic 
factor is remarkable because the recursive teaching dimension is 
known to be bounded by $\log |\cC|$ for any concept class $\cC$. 
\end{abstract}

\section{Introduction}

A tournament is said to have the $S_k$-property if, for any set 
of $k$ players, there is another player who beats them all. 
In the 1960's and early 1970's, several researchers pursued 
the goal of finding the smallest number $n$ such that there 
exists a tournament with $n$ players which has the $S_k$-property. 
In the year 1963, an upper bound on this number was proven 
by Erd\"os~\cite{E1963} by means of the probabilistic method. 
Eight years later, Graham and Spencer~\cite{GS1971} presented 
another upper bound, which is weaker than the bound proven 
by Erd\"os, but results from a concrete and easily constructible 
family of tournaments, namely the so-called quadratic-residue 
tournaments. 

The purpose of this paper is twofold. First we bring into
play a stronger version of the $S_k$-property. We demonstrate
that it is surprisingly simple to transfer the afore-mentioned
results to the new setting. Second, we show that the new setting 
has implications for teaching.  
For instance, the quadratic-residue tournaments induce an
infinite family of concept classes all of which can be taught
with a single example in the No-Clash model of teaching while,
in order to teach a class $\cC$ of this family in the recursive model
of teaching, order of $\log|\cC|$ many examples are required.
The family of concept classes induced by the quadratic-residue
tournaments is the first concrete and easily constructible family
which separates the two mentioned teaching models by more than
a constant factor. The existence of concept classes like this 
had been shown before~\cite{Sim2022} only by means of the probabilistic method.
The separation by a logarithmic factor is remarkable because the 
recursive teaching dimension is known~\cite{DFSZ2014} to be bounded 
by $\log |\cC|$ for any concept class $\cC$.

\medskip 
The paper is organized as follows. Section~\ref{sec:definitions}
fixes some notation and terminology. In Section~\ref{sec:sk-property},
it is shown that, for $k\ge1$, there exists a tournament of relative small
order which does have the strong $S_k$-property. In Section~\ref{sec:sk-construction},
the existence of such classes (though of somewhat larger order) is 
shown by construction. Here the QR-tournaments come into play.
The final Section~\ref{sec:implications2teaching} is devoted to 
the implications for teaching. Here the concept classes induced
by tournaments come into play.

\section{Tournaments and Teaching Models} \label{sec:definitions}

Section~\ref{subsec:tournaments} reminds the reader to the
definition of a tournament. The definition of the weak and
strong $S_k$-property is postponed to Section~\ref{sec:sk-property}.
The definition of QR-tournaments will be given in
Section~\ref{sec:sk-construction}. 
Section~\ref{subsec:teaching-models} calls into mind some models
of teaching, including the No-Clash and the recursive model.
The definition of a concept class induced by a tournament
is postponed to Section~\ref{sec:implications2teaching}.

\subsection{Tournaments} \label{subsec:tournaments}

A \emph{tournament $G = (V,E)$ of order $n$} is a complete 
oriented graph with $n$ vertices. In other words, $|V| = n$
and, for every choice of two distinct vertices $x,y \in V$,
exactly one of the edges $(x,y)$ and $(y,x)$ is contained in $E$.
Informally, we may think of $V$ as set of players who compete
against each other in pairs. An edge $(x,y) \in E$ can be 
interpreted as ``$x$ has beaten $y$''.

\subsection{Teaching Models} \label{subsec:teaching-models}

Readers familiar with teaching models may skip this section 
and proceed immediately to Section~\ref{sec:sk-property}.

\smallskip
A \emph{concept over domain $\cX$} is a function
from $\cX$ to $\{0,1\}$ or, equivalently, a subset of~$\cX$.
A set whose elements are concepts over domain $\cX$ is referred to
as a \emph{concept class over $\cX$}. The elements of $\cX$ are called
\emph{instances}. The powerset of $\cX$ is denoted by $\cP(\cX)$.

We now call into mind the definition of several popular teaching 
models and the corresponding teaching dimensions. The definition
of the No-Clash and the recursive model of teaching will later
help us to fully articulate the implications of our results
for teaching. From a technical point of view, we will later
require mainly the parameter $\TD_{min}$, which is specified
in the first part of the following definition:

\begin{definition}[Teaching Models~\cite{GK1995,ZLHZ2011,KSZ2019}]
\label{def:teaching-models}
Let $\cC$ be a concept class over $\cX$.
\begin{enumerate}
\item
A \emph{teaching set for $C \in \cC$} is a subset $D \seq \cX$
which \emph{distinguishes $C$ from any other concept in $\cC$}, i.e.,
for every $C' \in \cC\sm\{C\}$, there exists some $x \in D$
such that $C(x) \neq C'(x)$. The size of the smallest teaching set
for $C \in \cC$ is denoted by $\TD(C,\cC)$. The \emph{teaching
dimension of $\cC$ in the Goldman-Kearns  model of teaching} is
then given by
\[
\TD(\cC) = \max_{C \in \cC} |T(C,\cC)| \enspace .
\]
A related quantity is
\[
\TD_{min}(\cC) = \min_{C \in \cC} |T(C,\cC)| \enspace .
\]
\item
Let $T:\cC \ra \cP(\cX)$ be a mapping that assigns to every concept
in $\cC$ a set of instances. $T$ is called an 
\emph{NC-teacher\footnote{NC = No-Clash.} for $\cC$} 
if for every $C \neq C' \in \cC$, there exists $x \in T(C) \cup T(C')$
such that $C(x) \neq C'(c)$. The \emph{NC-teaching dimension of $\cC$}
is given by
\[ \NCTD(\cC) = \min \{\max_{C \in \cC} |T(C)|: T\mbox{ is an NC-teacher for $\cC$}\}
\enspace .
\]
\item
Let $\cC_{min} \seq \cC$ be the \emph{easiest-to-teach concepts in $\cC$}, i.e., 
\[
\cC_{min} = \{C \in \cC: \TD(C,\cC) = \TD_{min}(\cC)\} \enspace .
\]
The \emph{recursive teaching dimension of $\cC$} is then given by
\[
\RTD(\cC) = \left\{ \begin{array}{ll}
              \TD_{min}(\cC) & \mbox{if $\cC = \cC_{min}$} \\
              \max\{\TD_{min}(\cC) , \RTD(\cC\sm\cC_{min})\} &
              \mbox{otherwise}
            \end{array} \right.
\enspace .
\]
\end{enumerate}
\end{definition}

\noindent
Some remarks are in place here:
\begin{enumerate}
\item
It was shown in~\cite{DFSZ2014} that
\begin{equation} \label{eq:rtd-tdmin}
\RTD(\cC) = \max_{\cC' \seq \cC}\TD_{min}(\cC') \ge \TD_{min}(\cC)
\enspace .
\end{equation}
\item
The set $T(C)$ in Definition~\ref{def:teaching-models} is an
{\em unlabeled} set of instances. Intuitively, one should think
of the learner as receiving the \emph{correctly labeled} instances
i.e., the learner receives $T(C)$ {\em plus} the corresponding $C$-labels
where $C$ is the concept that is to be taught.
\item
We say that two concepts $C$ and $C'$ {\em clash} (with respect
to $T:\cC \ra \cP(\cX)$) if they \emph{agree on $T(C) \cup T(C')$}, i.e,
if they assign the same $0,1$-label to all instances in $T(C) \cup T(C')$.
NC-teachers for $\cC$ are teachers who avoid clashes between any pair
of distinct concepts from $\cC$.
\end{enumerate}

\section{Tournaments with the $S_k$-Property} \label{sec:sk-property}

In this paper, the $S_k$-property will be called ``weak $S_k$-property''
so that it can be easier distinguished from its strong counterpart.
Here are the formal definitions of the weak and the strong $S_k$-property:

\begin{definition}[Weak $S_k$-Property]
A tournament $G = (V,E)$ is said to have the \emph{weak $S_k$-property}
if the following holds: for any choice of $k$ distinct 
vertices $a_1,\ldots,a_k \in V$, there exists another vertex $x \in V$ 
such that $(x,a_j) \in E$ for $j=1,\ldots,k$.
\end{definition}

\begin{definition}[Strong $S_k$-Property]
A tournament $G = (V,E)$ is said to have the \emph{strong $S_k$-property}
if the following holds: for any choice of $k$ distinct 
vertices $a_1,\ldots,a_k \in V$ and any choice 
of $b_1,\ldots,b_k \in \{\pm1\}^k$, there exists another vertex $x \in V$ 
such that the following holds: 
\begin{equation} \label{eq:strong-sk}
\forall j=1,\ldots,k:
\left\{ \begin{array}{ll}
           (x,a_j) \in E & \mbox{if $b_j = +1$} \\
           (a_j,x) \in E & \mbox{if $b_j = -1$}
          \end{array} \right.  \enspace .
\end{equation}
\end{definition}

Let $f(k)$ (resp.~$F(k)$) be the smallest number $n \ge k$ such that 
there exists a tournament of order $n$ which has the weak 
(resp.~the strong) $S_k$-property. The following is known about
the function $f(k)$: 
\[
2^{k-1} (k+2) - 1 \le f(k) \le
\min\left\{n: \binom{n}{k} (1-2^{-k})^{n-k} < 1\right\} \le (1+o(1)) \ln(2) k^2 2^k
\enspace .
\]
The lower bound is found in~\cite{SS1965}. The upper bound
is from~\cite{E1963}. It is an easy application of the probabilistic 
method. Note that the gap between the lower and the upper bound is of 
order $k$. Clearly $f(k) \le F(k)$ so that each lower bound on $f(k)$
is a lower bound on $F(k)$ too. Moreover, an obvious application
of the probabilistic method yields an upper bound on $F(k)$ that
differs from the above upper bound on $f(k)$ only by inserting an additional
factor $2^k$ in front of $\binom{n}{k}$.\footnote{This factor accounts
for the possible choices of $b_1,\ldots,b_k$.} Hence we get
\[
2^{k-1} (k+2) - 1 \le F(k) \le 
\min\left\{n: 2^k \binom{n}{k} (1-2^{-k})^{n-k} < 1\right\} 
\le (1+o(1)) \ln(2) k^2 2^k \enspace .
\]

\section{Construction of Tournaments with the $S_k$-Property} 
\label{sec:sk-construction}

As outlined in Section~\ref{sec:sk-property}, the probabilistic method
yields good upper bounds on $f(k)$ or $F(k)$, however without
providing us with a concrete tournament which satisfies this bound.
As far as the function $f(k)$ is concerned,
Graham and Spencer~\cite{GS1971} have filled this gap.
They defined and analyzed a tournament that is is based on the quadratic 
residues and non-residues in the prime field $\Fp_p$. It became known 
under the name quadratic-residue tournament (or briefly QR-tournament): 

\begin{definition}[QR-Tournament] \label{def:qr-tournament}
Let $p$ be a prime such that $p \equiv 3 \pmod{4}$.
The \emph{QR-tournament} of order $p$ is the tournament $(V,E)$ 
given by
\[
V = \{0,1,\ldots,p-1\}\ \mbox{ and }\ E = \{(x,y) \in V \times V: 
\mbox{$x-y$ is a quadratic residue modulo $p$} \}
\]
\end{definition}

\noindent
Let $\chi: \Fp_p \ra \{-1,0,1\}$ be the function
\[
\chi(x) = \left\{ \begin{array}{rl}
      +1 & \mbox{if $x \neq 0$ is a quadratic residue modulo $p$} \\
      -1 & \mbox{if $x \neq 0$ is a quadratic non-residue modulo $p$} \\
       0 & \mbox{if $x=0$}
          \end{array} \right. \enspace .
\]

In the sequel, $p$ always denotes a prime that is congruent to $3$
modulo $4$ (so that $-1$ is a quadratic non-residue).
Note that the graph $(V,E)$ in Definition~\ref{def:qr-tournament}
is indeed a tournament because $\chi(y-x) = -\chi(x-y)$ so that
exactly one of the edges $(x,y)$ and $(y,x)$ is included in $E$.
Graham and Spencer have shown the following result:

\begin{theorem}[\cite{GS1971}] \label{th:qr-weak-sk}
The QR-tournament of order $p$ has the weak $S_k$-property
provided that $p > k^2 2^{2k-2}$.
\end{theorem}

\noindent
As we show now, the same construction works for the strong $S_k$-property:

\begin{theorem} \label{th:qr-strong-sk}
The QR-tournament of order $p$ has the strong $S_k$-property
provided that $p > k^2 2^{2k-2}$.
\end{theorem}

\begin{proof}
The proof will be a slight extension of the proof of Theorem~\ref{th:qr-weak-sk}
in~\cite{GS1971}, but it will have to deal with the 
variables $b_1,\ldots,b_k \in \{\pm1\}$ that occur in the definition 
of the strong $S_k$-property (and are missing in the definition 
of the weak $S_k$-property). \\
Consider first the case $k=1$. A tournament has the strong $S_1$-property
iff no vertex has in- or outdegree $p-1$. Every QR-tournament has this 
property because every vertex has in- and outdegree $\frac{p-1}{2} < p-1$. \\
The remainder of the proof is devoted to the case $k \ge 2$.
Let $G = (V,E)$ be the QR-tournament of order $p$. 
Let $a_1,\ldots,a_k \in V$ be $k$ distinct vertices 
and let $b_1,\ldots,b_k \in \{\pm1\}$. Set $A = \{a_1,\ldots,a_k\}$,
$a = (a_1,\ldots,a_k)$, $b = (b_1,\ldots,b_k)$ and consider the 
auxiliary functions
\[
g(a,b) = \sum_{x \in V \sm A} \prod_{j=1}^{k}[1 + b_j \chi(x-a_j)]\ 
\mbox{ and }\ 
h(a,b) = \sum_{x=0}^{p-1} \prod_{j=1}^{k}[1 + b_j \chi(x-a_j)] \enspace .
\]
An inspection of $g(a,b)$ reveals that there exists an $x \in V$ 
which satisfies~(\ref{eq:strong-sk}) if and only if $g(a,b) > 0$. 
It suffices therefore to show that $g(a,b) > 0$. 
To this end, we decompose $g(a,b)$ according to
\[
g(a,b) = p + (h(a,b)-p) - (h(a,b) - g(a,b)) \enspace .
\]
In order to show that $g(a,b) > 0$, it suffices to show that
\begin{equation} \label{eq:qr-ub}
|h(a,b)-p| \le \sqrt{p} \cdot ((k-2)2^{k-1}+1)\ \mbox{ and }\ h(a,b)-g(a,b) \le 2^k 
\end{equation}
because $p - \sqrt{p}((k-2)2^{k-1}+1) - 2^k > 0$ provided that $p > k^2 2^{2k-2}$,
as an easy calculation shows.\footnote{This calculation makes use of 
the case-assumption $k \ge 2$.} \\
We still have to verify~(\ref{eq:qr-ub}). In order 
to get $|h(a,b)-p| \le \sqrt{p}((k-2)2^{k-1}+1)$, we apply 
the distributive law and rewrite $h(a,b)$ as follows:
\begin{equation} \label{eq:h-expansion}
h(a,b) = \sum_{x=0}^{p-1} 1 + \sum_{x=0}^{p-1} \sum_{j=1}^{k}b_j \chi(x-a_j)
+ \sum_{r=2}^{k} S_r
\end{equation}
where
\begin{equation} \label{eq:higher-order-terms} 
S_r = \sum_{x=0}^{p-1} \sum_{1 \le j_1 <\ldots< j_r \le k}
\prod_{i=1}^{r}b_{j_i} \chi(x-a_{j_i}) =
\sum_{1 \le j_1 <\ldots< j_r \le k} \left(\prod_{i=1}^{r}b_{j_i}\right)
\sum_{x=0}^{p-1} \prod_{i=1}^{r} \chi(x-a_{j_i}) \enspace .
\end{equation}
Since $\sum_{x=0}^{p-1}1 = p$
and 
\[
\sum_{x=0}^{p-1} \sum_{j=1}^{k}b_j \chi(x-a_j) =
\sum_{j=1}^{k}b_j\underbrace{\sum_{x=0}^{p-1} \chi(x-a_j)}_{= 0} = 0,
\]
we can bring~(\ref{eq:h-expansion}) in the form
\[ h(a,b) - p = \sum_{r=2}^{k} S_r \enspace . \]
Burgess~\cite{B1962} has shown that
\[
\left|\sum_{x=0}^{p-1} \prod_{i=1}^{r} \chi(x-a_{j_i})\right| \le
(r-1)\sqrt{p}  
\]
holds for every fixed choice of $1 \le j_1 <\ldots< j_r \le k$. 
In combination with~(\ref{eq:higher-order-terms}), it follows that
\[ 
|h(a,b)-p| = \left|\sum_{r=2}^{k} S_r\right| \le 
\sqrt{p} \cdot \sum_{r=2}^{k} \binom{k}{r}(r-1) \enspace .
\]
A straightforward calculation shows 
that $\sum_{r=2}^{k} \binom{k}{r}(r-1) = (k-2)2^{k-1}+1$.
We may therefore conclude that the first inequality in~(\ref{eq:qr-ub}) 
is valid. We finally have to show that $h(a,b)-g(a,b) \le 2^k$. 
Note first that
\[ 
h(a,b) - g(a,b) = \sum_{i=1}^{k} \prod_{j=1}^{k} [1 + b_j \chi(a_i-a_j)]
\enspace .
\]
We call $\prod_{j=1}^{k} [1 + b_j \chi(a_i-a_j)]$ the \emph{contribution}
of $i$ to $h(a,b)-g(a,b)$. Set $I_b = \{i\in\{0,1,\ldots,p-1\}: b_i = b\}$ 
for $b=\pm1$.  The following observations are rather obvious:
\begin{itemize}
\item
Every $i$ makes a contribution of either $0$ or $2^{k-1}$.
\item
If $i$ makes a non-zero contribution, then $\chi(a_i-a_j) = b_j$
for every $j \neq i$.
\item
For each $b \in \{\pm1\}$, at most one $i \in I_b$ makes a non-zero
contribution.\footnote{If two distinct $i,i' \in I_b$ made a
non-zero contribution, then we would 
get $\chi(a_i-a_{i'}) = b = \chi(a_{i'}-a_i)$, which is in
contradiction to $\chi(a_{i'}-a_i) = -\chi(a_i-a_{i'})$.}
\end{itemize}
These observations imply that $h(a,b) - g(a,b) \le 2^k$,
which concludes the proof of the theorem.
\end{proof}

\section{Implications for Teaching} \label{sec:implications2teaching}

With each tournament $G = (V,E)$, we associate the 
concept class $\cC(G) = \{C_x: x \in V\}$ given by
\[ C_x = \{a \in V: (x,a) \in E\} \enspace . \] 
Intuitively, we can think of $x$ as a player in the tournament
and of $C_x$ as the set of players who were beaten by $x$.
Note that$\TD_{min}(G) \le k$ intuitively means that there exists 
a player $x \in V$ who can be uniquely identified from telling which 
of $k$ (appropriately chosen) players he has beaten, and which he has 
not beaten.

\noindent
It is well known that concept classes induced by a tournament are 
easy to teach in the NC-model:
\begin{remark}[\cite{Sim2022}]
For every tournament $G$, we have that $\NCTD(\cC(G)) \le 1$
(with equality for all tournaments of order at least $2$).
\end{remark}

Let $G$ be a tournament of order $n$.
Since, as noted already in Section~\ref{subsec:teaching-models},
the recursive teaching dimension is lower bounded by $\TD_{min}$,
we can show that $\RTD(\cC(G))$ exceeds $\NCTD(\cC)$ by a 
factor of order $\log(n)$ by proving logarithmic lower bounds
on $\TD_{min}(\cC(G))$ for appropriately chosen tournaments $G$.
This is precisely what we will do in the sequel.

We first relate the parameter $\TD_{min}(G)$ to a refinement of the
strong $S_k$-property. A tournament $G = (V,E)$ is said to have
the \emph{strong $S_{k,m}$-property} if the following holds:
for any choice of $k$ distinct vertices $a_1,\ldots,a_k$
and any choice of $b_1,\ldots,b_k \in \{\pm1\}$, there
exists a set $X \seq V$ of size $m$ such that every $x \in X$
satisfies~(\ref{eq:strong-sk}). Let $F(k,m)$ be the smallest 
number $n$ such that there exists a tournament of order $n$
which has the strong $S_{k,m}$-property. It is rather obvious 
that the following holds:

\begin{remark}
\begin{enumerate}
\item
If $G$ has the strong $S_{k,2}$-property, then $\TD_{min}(G) > k$.
\item
The strong $S_{k,1}$-property coincides with the strong $S_k$-property.
Consequently $F(k) = F(k,1)$.
\item
The strong $S_{k,m+1}$-property implies the strong $S_{k,m}$-property.
Consequently $F(k,m) \le F(k,m+1)$.
\item
The strong $S_{k+1}$-property implies the strong $S_{k,2}$-property.
Consequently $F(k,2) \le F(k+1)$.\footnote{Using methods from~\cite{SS1965},
it can even be shown that the strong $S_{k+1}$-property implies 
the strong $S_{k,k+2}$-property. Consequently $F(k,k+2) \le F(k+1)$.}
\item
$F(k,m)$ is non-decreasing in both arguments.
\end{enumerate}
\end{remark}

\noindent
The following is an immediate consequence of the first of these remarks;

\begin{corollary} \label{cor:lb1-tdmin}
Let $G_k$ denote a tournament having the strong $S_{k,2}$-property 
and being of order $F(k,2)$. Then $\TD_{min}(\cC(G_k)) > k$.   
\end{corollary}

Corollary~\ref{cor:lb1-tdmin} does not tell explicitly 
how $\TD_{min}(\cC(G_k))$ depends on $n = F(k,2)$.
As already shown in~\cite{Sim2022}, a more useful lower bound 
can be obtained via the probabilistic method: 

\begin{theorem}[\cite{Sim2022}]
For every sufficiently large $n$, there exists a tournament $G_n$
of order $n$ such that
\[
\TD_{min}(\cC(G_n)) > \log(n) - 2 \log\log(2n) - 2 \enspace .
\]
\end{theorem}

\begin{proof}
The proof given here makes use of the 
implication ``$\mbox{strong }S_{k+1} \impl \mbox{strong }S_{k,2}$'' 
and is slightly simpler than the proof given in~\cite{Sim2022}. If 
\begin{equation} \label{eq:prob-method}  
2^{k+1} \binom{n}{k+1} (1-2^{-(k+1)})^{n-(k+1)} < 1 \enspace ,
\end{equation}
then there is a strictly positive probability for the event
that a random tournament of order $n$ has the strong $S_{k+1}$- and
therefore also the strong $S_{k,2}$-property. In this case, we may
conclude that there exists a tournament $G_n$ of order $n$
such that $\TD_{min}(\cC(G_n)) > k$. We may clearly assume 
that $n \ge 2(k+1)$ so that $n-(k+1) \ge n/2$. Making use of $n-(k+1) \ge n/2$, 
$\binom{n}{k} \le n^k$ and $1+x \le e^x$ with equality for $x=0$ only, 
we get the following sufficient condition for~(\ref{eq:prob-method}):
\[
(2n)^{k+1} \exp\left(-\frac{n}{2^{k+2}}\right) \le 1 \enspace .
\]
After taking logarithm on both hand-sides and rearranging some terms,
this becomes
\[
(k+1) 2^{k+2} \ln(2n) \le n \enspace .
\]
A straightforward calculation shows that the latter condition
is satisfied whenever $k \le \log(n) - 2 \log\log(2n) - 2$. 
From this discussion, the assertion of the theorem is immediate.
\end{proof}

Our main implication for teaching is the fact that the
quadratic-residue tournament of order $p$ induces
a concept class whose $\TD_{min}$ grows logarithmically
with $p$:

\begin{theorem}
Let $p$ be a prime that is congruent to $3$ modulo $4$.
Let $G_p$ be the quadratic-residue tournament of order $p$.
Then 
\[
\TD_{min}(G_p) > \frac{1}{2} \log(p) - \log\log(p) - 1 \enspace . 
\]
\end{theorem}

\begin{proof}
We make again use of the 
implication ``$\mbox{strong }S_{k+1} \impl \mbox{strong }S_{k,2}$''.
According to Theorem~\ref{th:qr-strong-sk}, the following holds:
if 
\begin{equation} \label{eq:qr-lb}
(k+1)^2 2^{2k} < p \enspace ,
\end{equation} 
then $G_p$ has the strong $S_{k+1}$- and therefore also the strong $S_{k,2}$-property. 
In this case, we may conclude that $\TD_{min}(G_p) > k$.
A straightforward calculation shows that (\ref{eq:qr-lb}) holds
whenever $k \le \frac{1}{2} \log(p) - \log\log(p) - 1$.
From this discussion, the assertion of the theorem is immediate.
\end{proof}



\end{document}